%% file: paper.tex
\documentclass[runningheads]{llncs}
\input{preamble}
\usepackage[T1]{fontenc}
\usepackage{color}

\begin{document}
\title{Minimizing $\ell_2$ Norm of Flow Time by Starvation Mitigation}
%
%
\author{Tung-Wei Kuo\orcidID{0000-0002-2518-4462}}
\authorrunning{T.-W. Kuo}
\institute{
Department of Computer Science, National Chengchi University, Taipei, Taiwan
\email{twkuo@cs.nccu.edu.tw}}
\maketitle

\begin{abstract}
The assessment of a job's Quality of Service (QoS) often revolves around its flow time, 
also referred to as response time. This study delves into two fundamental objectives 
for scheduling jobs: the average flow time and the maximum flow time. 
While the Shortest Remaining Processing Time (SRPT) algorithm minimizes average flow time, 
it can result in job starvation, causing certain jobs to experience disproportionately long 
and unfair flow times. In contrast, the First-Come-First-Served (FCFS) algorithm minimizes the 
maximum flow time but may compromise the average flow time. 

To strike a balance between these two objectives, a common approach
is to minimize the $\ell_2$ norm of flow time. 
SRPT and FCFS are $O(n^{\frac{1}{2}})$-competitive for this problem, where $n$ is the number of jobs.
Prior to this work, no algorithm is known to achieve a competitive ratio better than 
SRPT and FCFS. In this paper, we use FCFS to mitigate the starvation caused by SRPT.  
Given a good estimate of $n$, we prove that this approach achieves a much better competitive 
ratio of $O(n^{\frac{1}{3}})$.  Our results provide the first theoretical evidence that mitigating 
starvation in SRPT leads to a provable improvement in scheduling performance.

\end{abstract}

\input{intro}

\input{alg}
\input{best}
\input{thrmN}
\input{thrmS}

\section{Concluding Remark}
While numerous studies have explored strategies to mitigate the starvation issue inherent in SRPT scheduling 
\cite{bai2017pias,10.1007/BFb0037157,10.1145/2987550.2987563,9269382,10.1117/12.538820,646697,7018931,9521259}, 
no theoretical work has rigorously established whether such mitigation efforts lead to meaningful improvements. 
In this study, we address this gap by analyzing the competitive ratio of the $\ell_2$ 
norm of flow time. Our findings provide the first theoretical evidence that starvation mitigation 
significantly enhances SRPT performance.

\begin{credits}
\subsubsection*{Acknowledgments.}
This work was supported in part by the Ministry of Science and Technology of Taiwan 
(MOST 111-2221-E-004-003-MY2) and the National Science and Technology Council of Taiwan 
(NSTC 113-2221-E-004-011-MY2). 
The author is grateful to the anonymous reviewers for their valuable suggestions, 
particularly regarding the refinement of Lemma~\ref{lemma: LB1} and its proof.
\end{credits}

\bibliographystyle{splncs04}
\bibliography{paper}

\begin{appendix}
\input{CR}
\input{thrmS2}
\end{appendix}

\end{document}

%% file: preamble.tex
\usepackage{amssymb}
\usepackage{amstext}
\usepackage{amsmath}
\usepackage{thm-restate}
\usepackage{subcaption}
\usepackage{pgfplots}
\pgfplotsset{
    compat=1.6,
    table/search path={data},
}
\usetikzlibrary{matrix}
\usepgfplotslibrary{groupplots}

\pgfplotscreateplotcyclelist{mylist}{
    {red!60!black,mark=otimes*, mark options={fill=red!40}},    
    {red,mark=triangle*},
    {red,mark=square},
    {blue!60!black,mark=otimes*, mark options={fill=blue!40}},
    {blue,mark=triangle*},
    {blue,mark=square},
    {black}
}

\usepackage{graphicx}
\usepackage[title]{appendix}
\usepackage{textcomp}
\usepackage{xcolor}
\usepackage{algorithm}
\usepackage[noend]{algorithmic}
\usepackage{etoolbox}
\usepackage{mathtools}
\usepackage{environ}
\usepackage{xfrac}
\usepackage{booktabs}
\usepackage{multirow}
\usepackage{hyperref}

\hypersetup{
    colorlinks = true,
    linkcolor= red,
    citecolor= blue,
    urlcolor=blue
}

\usepackage{amsthm}

\DeclareMathOperator*{\argmin}{arg\,min}
\DeclareMathOperator*{\argmax}{arg\,max}
\DeclareMathOperator{\dom}{dom}

%% file: intro.tex
\section{Introduction}\label{sec: intro}
\begin{figure}[t]
\begin{center} 
\includegraphics[width= 8 cm]{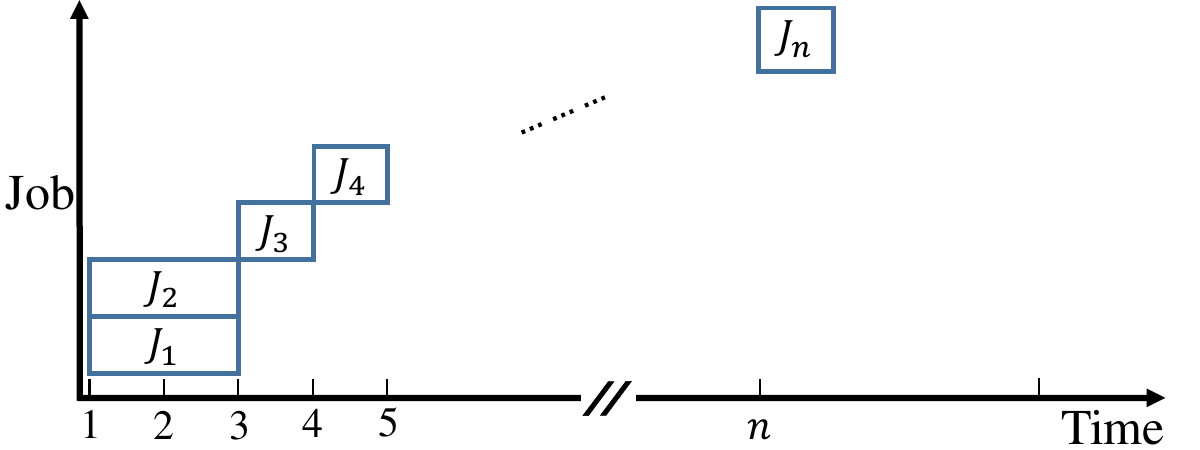} 
\caption{A bad instance for SRPT.}
\label{fig: ex1}
\end{center}
\end{figure}

\begin{figure}[h!]
\begin{center} 
\includegraphics[width= 8 cm]{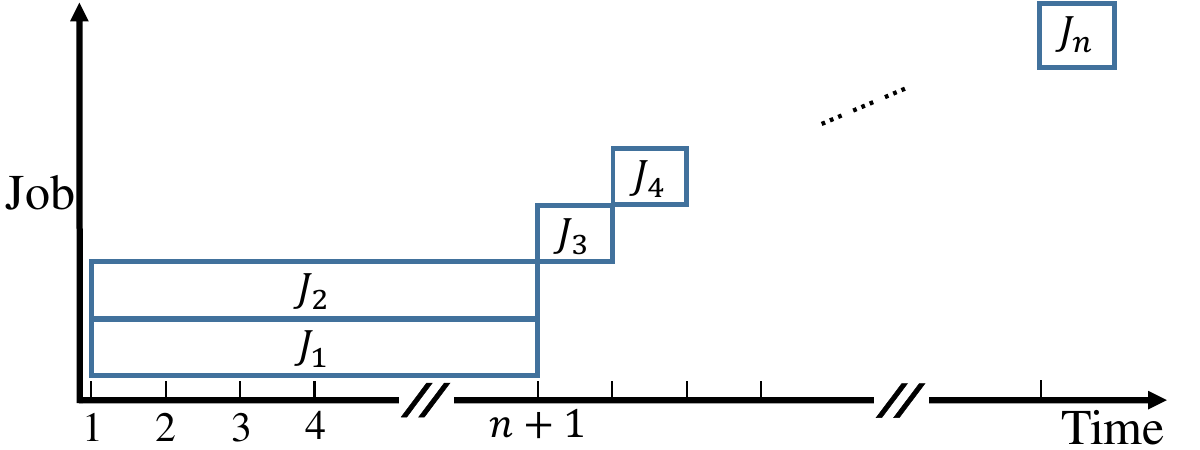} 
\caption{A bad instance for FCFS.}
\label{fig: ex2}
\end{center}
\end{figure}
In a server system, the Quality of Service (QoS) of a job is normally measured by 
its response time, which is defined as the amount of time between 
job release and job completion. In job scheduling, the response time is also called the flow time. 
Specifically, for any schedule $\mathcal{S}$, let $c_i(\mathcal{S})$ be the 
completion time of job $J_i$ under $\mathcal{S}$ and $r_i$ be the release time of $J_i$. 
Then the flow time of a job $J_i$ under $\mathcal{S}$, denoted by $f_i(\mathcal{S})$, 
is defined as
\[
f_i(\mathcal{S}) = c_i(\mathcal{S}) - r_i.
\]
Two natural QoS objectives for a server system are the average, 
or equivalently, $\ell_1$ norm, 
of job flow times, and the maximum, or $\ell_{\infty}$ norm, of job flow times. 
It is well-known that Shortest Remaining Processing Time (SRPT), which minimizes average flow time, 
may cause some jobs to have long and inequitable flow times (i.e., job starvation).
In contrast, First-Come-First-Served (FCFS), 
which minimizes the maximum flow time~\cite{10.5555/314613.314715}, 
may deteriorate the average flow time. 

\begin{example}
In Fig.~\ref{fig: ex1}, jobs $J_1$ and $J_2$ are released 
at time 1, and both have processing times of 2. All the other jobs $J_i$ 
are released at time $i$ and have unit processing times. Assume that under 
SRPT, $J_1$ is executed first. Thus, $J_2$ is starving in the sense that 
it has to wait until all the other jobs are completed. Under FCFS, 
however, the flow time of every job is only $O(1)$. 
\end{example}

\begin{example} Consider Fig.~\ref{fig: ex2}.
Compared with Fig.~\ref{fig: ex1}, the processing times of $J_1$ and 
$J_2$ in Fig.~\ref{fig: ex2} are increased to $n$. Under FCFS, 
the flow time of every job is $\Theta(n)$, and thus FCFS's average 
flow time is $\Theta(n)$. Under SRPT, however, the average flow time is only $O(1)$. 
\end{example}

The standard approach to balance the $\ell_1$ norm and the $\ell_{\infty}$ norm of flow time is 
minimizing the $\ell_2$ norm of flow time~\cite{doi:10.1137/090772228,10.1145/2755573.2755581}. 
Specifically, the $\ell_2$ norm of flow time 
under schedule $\mathcal{S}$ is 
\[
\sqrt{f_1(\mathcal{S})^2 + f_2(\mathcal{S})^2 + \cdots + f_n(\mathcal{S})^2}, 
\]
where $n$ is the number of jobs.
In this paper, we aim to minimize the $\ell_2$ norm of flow time online 
on a single machine. 

For the online problem of minimizing the $\ell_2$ norm of flow time, 
only standard algorithms, including SRPT, 
Shortest Job First (SJF), Shortest Elapsed Time First (SETF),
and Round Robin (RR), have been analyzed, 
assuming that the machine is gifted extra speed~\cite{doi:10.1137/090772228,10.1145/2755573.2755581}. 
However, SRPT, SJF, and SETF may cause job starvation, 
and FCFS and RR may deteriorate 
average flow time. 
These drawbacks worsen the $\ell_2$ norm of flow time (see Appendix~\ref{sec: cr}). 
For instance, the examples illustrated in Figs.~\ref{fig: ex1} and \ref{fig: ex2} 
demonstrate that FCFS and SRPT are $\Omega(n^{1/2})$-competitive  
for minimizing the $\ell_2$ norm of flow time.

In this paper, we use FCFS to mitigate the starvation caused by SRPT. 
Our algorithm follows SRPT initially. 
If a job has been in the system for too long
and its remaining processing time is small enough, then the job 
becomes a starving job. 
Specifically, a job $J_i$ becomes starving at time $t$ if 
\[
\frac{t-r_i}{\text{remaining processing time of } J_i} \geq \theta,
\]
where $\theta$ is called the \textit{starvation threshold}. 
Whenever some jobs are starving, 
the algorithm processes the job that becomes starving first. 
We use $BAL(\theta)$ to denote the algorithm when 
the starvation threshold is set to $\theta$. 
For example, $BAL(0)$ is equivalent to FCFS and $BAL(\infty)$ 
is equivalent to SRPT.

\subsection{Our Result}
Our analysis on $BAL(\theta)$'s competitive ratio suggests that $\theta = n^{\frac{2}{3}}$ 
is the best starvation threshold.
Among all the aforementioned standard scheduling algorithms, 
SRPT and FCFS achieve the best known competitive ratio of $O(n^{\frac{1}{2}})$ for 
minimizing the $\ell_2$ norm of flow time (see Appendix~\ref{sec: cr}). 
We prove that $BAL(n^{\frac{2}{3}})$, a combination of SRPT and FCFS, achieves a significantly 
better competitive ratio of $O(n^{\frac{1}{3}})$. 

In practice, we can only estimate $n$. 
For example, usually job scheduling is only critical during peak hours,  
and normally the job arrival process is modeled as a Poisson process. 
Thus, the product of the job arrival 
rate and the duration of peak hours would be a good estimate of $n$. 
Nonetheless, in this study, we do not make any assumptions about the 
job arrival process or the estimation method. 
Instead, we assume that an estimate $\tilde{n}$ of $n$ with bounded error is given. 
The following theorem gives the competitive ratio of $BAL(\tilde{n}^{\frac{2}{3}})$. 
Specifically, if $\tilde{n} = \Theta(n)$, then $BAL(\tilde{n}^{\frac{2}{3}})$ is 
$O(n^{\frac{1}{3}})$-competitive.

\begin{restatable}{theorem}{beat}\label{coro: beat}
Let $\tilde{n}$ be an estimate of $n$ such that 
$\beta n \leq \tilde{n} \leq \alpha n$ for 
some $\frac{1}{n} \leq \beta \leq 1$ and $\alpha \geq 1$. 
Then the competitive ratio of $BAL(\tilde{n}^{\frac{2}{3}})$ 
for minimizing the $\ell_2$ norm of flow time is 
$O\left(\tilde{n}^{\frac{1}{3}} + n^{\frac{1}{2}}/\tilde{n}^{\frac{1}{6}}\right)
= O\left(
n^{\frac{1}{3}} 
\left(
\alpha^{\frac{1}{3}} + \beta^{\frac{-1}{6}}
\right)
\right)$.
\end{restatable} 

Observe that if the predicted number of jobs $\tilde{n}$ overestimates the true number $n$, 
$BAL(\tilde{n}^{\frac{2}{3}})$ achieves an $O(\tilde{n}^{\frac{1}{3}})$-competitive ratio, which corresponds to the 
usual competitive ratio of the algorithm, but expressed in terms of the predicted number of jobs 
rather than the actual count.
Conversely, if the prediction underestimates the true count, the competitive ratio is $O(n^{\frac{1}{2}}/\tilde{n}^{\frac{1}{6}})$,
which is no worse than that of standard algorithms without prediction.

\subsection{Related Work}\label{subsec: related}
In their seminal work, Bansal and Pruhs introduced 
the online problem of minimizing the $\ell_p$ norm of flow 
time~\cite{doi:10.1137/090772228}. They proved that 
SRPT and SJF are $(1+\epsilon)$-speed $O(\frac{1}{\epsilon})$-competitive 
and that SETF is $(1+\epsilon)$-speed $O(\frac{1}{\epsilon^{2+2/p}})$-competitive.
The results have been extended to all symmetric norms of 
flow time~\cite{golovin_et_al:LIPIcs:2008:1753} and identical 
machines~\cite{10.1145/1007352.1007411,10.5555/2133036.2133046}.  
Moreover, Im et al. showed that RR is $O(1)$-speed $O(1)$-competitive 
for the $\ell_2$ norm of flow time~\cite{10.1145/2755573.2755581}.
In~\cite{angelopoulos2019primal,doi:10.1137/120902288}, 
more general objective functions were considered. 
Specifically, for a job $J_i$ with flow time $f_i$, 
a cost $g_i(f_i)$ is incurred. The only restriction on $g_i$ is 
that $g_i$ must be non-decreasing. 
For this general cost minimization problem, there are $O(1)$-speed 
$O(1)$-competitive 
online algorithms~\cite{angelopoulos2019primal,doi:10.1137/120902288}. 

For the offline setting, Bansal and Pruhs first proposed an $O((\log \log P)^{1/p})$-approximation 
algorithm using linear programming rounding, where $P$ is the ratio of the maximum to minimum job 
size~\cite{doi:10.1137/130911317}. 
Subsequently, several $O(1)$-approximation algorithms were developed, 
leveraging linear programming rounding, dynamic programming, or a combination of both~\cite{doi:10.1137/1.9781611977936.7,doi:10.1137/19M1244512,doi:10.1137/18M1202451}.

%% file: alg.tex
\section{Definitions and the Algorithm}\label{sec: thealgo}
We consider $n$ jobs, $J_1, J_2, \cdots, J_n$, 
and one machine. Each job $J_i$ has a processing time $p_i$ and a release time $r_i$. 
As in~\cite{doi:10.1137/130911317,doi:10.1137/18M1202451}, 
we assume that $p_i$ and $r_i$ are integers. 
We allow job preemption and consider clairvoyant scheduling.
Define $F(\mathcal{S}) = \sum_{i=1}^{n}{f_i(\mathcal{S})^2}$. 
The goal is to compute online a schedule $\mathcal{S}$ that minimizes 
the $\ell_2$ norm of flow time, i.e., $\sqrt{F(\mathcal{S})}$. 

For every $t \in \mathbb{N}$\footnote{In this paper, 
we assume $\mathbb{N}$ contains 0.}, the \textbf{time slot} $[t]$ 
is defined as the time interval between time $t$ and time $t+1$. 
Thus, we can divide time into time slots $[0], [1], [2], \cdots$.
We can view each job $J_i$ as a chain of tasks 
$J_{i,1}, J_{i,2}, \cdots, J_{i,p_i}$, where each task has a unit 
processing time. Because all the processing times and release times are integers, 
by a simple exchange argument, we can assume that under an optimal schedule,
the machine never executes more than one task in time slot $[t]$ for any 
$t \in \mathbb{N}$ (i.e., the machine is either idle or executing 
the same task throughout the entire time slot $[t]$). 
Thus, for every time slot $[t]$, a schedule assigns a (possibly empty) task to be executed in $[t]$.
If a task $J_{i,k}$ is executed in time slot $[t]$ under schedule $\mathcal{S}$, 
then its completion time, denoted by $c_{i,k}(\mathcal{S})$, is $t+1$. 
Throughout this paper, we use $\mathcal{BAL}(\theta), \mathcal{SRPT}, \mathcal{FCFS}, \mathcal{RR}$, 
and $\mathcal{OPT}$ to denote the schedule obtained by $BAL(\theta)$, 
SRPT, FCFS, RR, and an optimal schedule, respectively. 
If $\theta$ is clear from the context, we simply write $\mathcal{BAL}$ instead of $\mathcal{BAL}(\theta)$.

A job is said to be \textbf{active} at time $t$ under schedule $\mathcal{S}$ 
if it is released by time $t$ but has not yet been completed by time $t$ 
under $\mathcal{S}$. We use $A(t, \mathcal{S})$ to denote the index set of the 
active jobs under $\mathcal{S}$ at time $t$. 
In this paper, for any $a, b \in \mathbb{N}$ with $a \leq b$, $[a,b]$ is defined as 
$\{i | i \in \mathbb{N}, a \leq i \leq b\}$. If $a > b$, then $[a,b] = \varnothing$. 
Moreover, we define a \textbf{map} as a non-negative function with a finite domain.   
In this paper, we use the terms ``remaining processing time'' and ``number of remaining tasks'' 
interchangeably. 
We then introduce the most important map in this paper. 
\begin{definition}
For any schedule $\mathcal{S}$, any time $t$, and any $i \in A(t, \mathcal{S})$, 
define $q_{t, \mathcal{S}}(i)$ as the number of remaining tasks of $J_i$ 
at time $t$ under $\mathcal{S}$; if $i \in [1,n] \setminus A(t, \mathcal{S})$, 
define $q_{t, \mathcal{S}}(i) = 0$.
\end{definition}
%
%

\begin{example}
Consider the schedule obtained by SRPT for the instance shown 
in Fig.~\ref{fig: exq}. 
Assume that SRPT first executes $J_1$. Then, at time 3, the number of remaining 
tasks of $J_1$ and $J_2$ are 4 and 6, respectively. 
Thus, $q_{3, \mathcal{SRPT}}(1) = 4$ and $q_{3, \mathcal{SRPT}}(2) = 6$. 
Because jobs $J_3, J_4, J_5$ are not active 
at time 3, we have $q_{3, \mathcal{SRPT}}(i) = 0$ for all $3 \leq i \leq 5$. 
\end{example}

\begin{figure}[t]
\begin{center} 
\includegraphics[width=8cm]{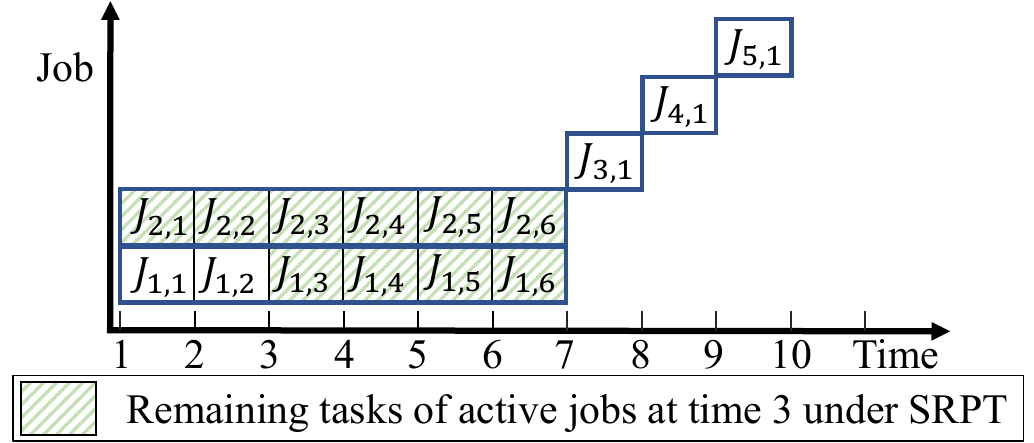} 
\caption{An example of $q_{3, \mathcal{SRPT}}(i)$.}
\label{fig: exq}
\end{center}
\end{figure}

\subsection{The Algorithm}
Next, we formalize the algorithm $BAL(\theta)$ introduced in Section~\ref{sec: intro}. 
For brevity, define $q_t(i)$ as $q_{t, \mathcal{BAL}}(i)$. In other words, 
$q_t(i)$ is the number of remaining tasks of $J_i$ at time $t$ under $BAL(\theta)$. 
$BAL(\theta)$ categorizes active jobs into two types, starving and normal, according to the 
given starvation threshold $\theta$. 
Every job is \textbf{normal} initially, and tasks of normal jobs are called normal tasks.
An active job $J_i$ is said to be \textbf{starving} at time $t$ if 
\begin{equation}\label{eq: stthreshold}
\frac{t-r_i}{q_t(i)} \geq \theta.
\end{equation}
Moreover, once $J_i$ becomes starving, every remaining task of $J_i$ is said to 
be starving as well. For any starving job $J_i$, 
define $t_i$ as the time at which $J_i$ becomes starving (i.e.,
$t_i$ is the smallest $t \in \mathbb{N}$ that satisfies Eq.~\eqref{eq: stthreshold}). 

When there are starving jobs, $BAL(\theta)$ executes the starving job that has the smallest $t_i$. 
In other words, $BAL(\theta)$ executes the job that becomes starving first. 
In Algorithm~\ref{algo}, $ST$ denotes the set of the starving job indices.

\begin{algorithm}[t]
\begin{small}
\caption{$BAL(\theta)$}
\label{algo}
\begin{algorithmic}[1]
\STATE{$ST \gets \varnothing$}
\FOR {$t \gets 0$ to $\infty$}
	\FOR{every active job $J_i$}
		\IF{$i \notin ST \land \frac{t-r_i}{q_t(i)} \geq \theta$}
            \STATE{$t_i \gets t$}            
            \STATE{Add $i$ to $ST$}
		\ENDIF
	\ENDFOR
	\IF{$ST \neq \varnothing$}
        \STATE{$i^* = \argmin_{i \in ST}{t_i}$}        	
    	\STATE{Execute job $J_{i^*}$ in time slot $[t]$}
        \STATE{Remove $i^*$ from $ST$ if $J_{i^*}$ is completed at time $t+1$}    	
	\ELSIF {$M_{SRPT}(t) \setminus M_{BAL}(t-1) \neq \varnothing$}
        \STATE{Among all the tasks in $M_{SRPT}(t) \setminus M_{BAL}(t-1)$, 
        execute the one that has the smallest completion time under SRPT}
    \ENDIF
\ENDFOR
\end{algorithmic}
\end{small}
\end{algorithm}  

When there are no starving jobs, $BAL(\theta)$ compensates normal jobs for their lost execution time
(as compared to SRPT). More precisely, during the execution of a starving job by $BAL(\theta)$, 
there is an opportunity for a normal job to be executed under SRPT, 
resulting in a delay compared to SRPT for this particular job. 
To address this delay, when there are no starving jobs, 
$BAL(\theta)$ gives priority to these delayed normal jobs. 
In particular, among all the delayed normal tasks, $BAL(\theta)$ executes the task 
that has the earliest completion time under SRPT, as it has experienced the longest delay. 
Specifically, let $M_{SRPT}(t)$ and $M_{BAL}(t)$ 
be the sets of tasks that are executed in time slots $[0], [1], [2], \cdots, [t]$
under SRPT and $BAL(\theta)$, respectively. Define $M_{BAL}(-1) = \varnothing$.
At time $t$, if there are no starving jobs, then among the tasks in $M_{SRPT}(t) \setminus M_{BAL}(t-1)$, 
$BAL(\theta)$ executes the one that is completed the earliest under SRPT.\footnote{All 
the variables associated with $BAL(\theta)$ depend on $\theta$ (e.g., $t_i$, $ST$, and 
$M_{BAL}(t)$). For simplicity, we do not explicitly show their dependency on $\theta$ in our notation.}   


%% file: best.tex
\section{Searching for the Best $\theta$: Proof of Theorem~\ref{coro: beat}}\label{appendix: thrm: algoN}
To find the best starvation threshold $\theta$, we partition the lifetime of $J_i$ into two periods, 
normal and starving. Specifically, if a job never becomes starving under $BAL(\theta)$, 
it is termed a \textbf{Finished-as-Normal (FaN) job}; 
otherwise, it is termed a \textbf{Finished-as-Starving (FaS) job}. 
We further extend the definition of $t_i$ to FaN jobs as follows: 
If $J_i$ is an FaN job, $t_i$ is defined as its completion time 
(i.e., $t_i = c_i(\mathcal{BAL})$). 
For any job $J_i$, we then define 
\[
norm_i(\theta) = t_i - r_i
\]
and
\[
starv_i(\theta) = c_i(\mathcal{BAL}) - t_i.
\]
Thus, $f_i(\mathcal{BAL}) = norm_i(\theta)+starv_i(\theta)$
and 
\[ 
F(\mathcal{BAL}) = \Theta\left(\sum_{i=1}^{n}{norm_i(\theta)^2}+\sum_{i=1}^{n}{starv_i(\theta)^2}\right).
\]

To minimize the $\ell_2$ norm of flow time, 
we balance $\sum_{i=1}^{n}{norm_i(\theta)^2}$ 
and $\sum_{i=1}^{n}{starv_i(\theta)^2}$. 
Intuitively, $norm_i(\theta)$ decreases as the starvation threshold $\theta$ decreases. 
In contrast, $starv_i(\theta)$ decreases as the starvation threshold $\theta$ increases. 
The following two theorems, whose proofs are postponed to Sections~\ref{sec: algoN} and \ref{sec: algoS}, 
relate $\sum_{i=1}^{n}{norm_i(\theta)^2}$ and $\sum_{i=1}^{n}{starv_i(\theta)^2}$ to $F(\mathcal{OPT})$. 
\begin{restatable}{theorem}{mainN}\label{thrm: algoN}
For any starvation threshold $\theta \geq 0$, 
$\displaystyle \sum_{i=1}^{n}{norm_i(\theta)^2} = O(\theta)F(\mathcal{OPT})$.
\end{restatable}

\begin{restatable}{theorem}{mainS}\label{thrm: algoS}
For any starvation threshold $\theta > 0$, 
$\displaystyle \sum_{i=1}^{n}{starv_i(\theta)^2} 
= O\left(\frac{n}{\sqrt{\theta}}\right)F(\mathcal{OPT})$.
 
\end{restatable}

Solving $\theta = \frac{n}{\sqrt{\theta}}$ yields $\theta = n^{\frac{2}{3}}$.
Thus, Theorems~\ref{thrm: algoN}~and~\ref{thrm: algoS} suggest that, 
to minimize the $\ell_2$ norm of flow time, 
the best starvation threshold is $\theta = n^{\frac{2}{3}}$, leading to a competitive ratio of 
$O(n^{\frac{1}{3}})$. In addition, Theorem~\ref{coro: beat} is a direct result of
Theorems~\ref{thrm: algoN}~and~\ref{thrm: algoS}. We give the proof of 
Theorem~\ref{coro: beat} in Appendix~\ref{appendix: coro: beat}.

%% file: thrmN.tex
\section{Proof of Theorem~\ref{thrm: algoN}}\label{sec: algoN}
To minimize the $\ell_2$ norm of flow time, we should not waste any time slot. Specifically, at every time 
$t \in \mathbb{N}$, if some active jobs are not completed, then one of them must 
be executed in time slot $[t]$. A schedule that satisfies the above property is 
called a \textbf{work-conserving} schedule. 
Clearly, for any time $t$, all work-conserving schedules have the same total number of remaining tasks 
over all active jobs.  
Thus, we define $q(t)$ as the total number of remaining tasks over all active jobs at time $t$ under any 
work-conserving schedule. For example, consider the instance shown in Fig.~\ref{fig: exq} at time $t = 3$. 
We then have $q(t) = 10$.
The following lemma gives the relationship between $q(t)$ and $F(\mathcal{OPT})$.

\begin{lemma}\label{prop: LBQ} 
For any work-conserving schedule $\mathcal{S}$,
$\displaystyle 
F(\mathcal{OPT}) 
\geq \sum_{t \in \mathbb{N}}{q(t)} 
= \sum_{i=1}^{n}{\sum_{t = r_i}^{c_i(\mathcal{S})}{q_{t, \mathcal{S}}(i)}}$.
\end{lemma}

\begin{proof}
Observe that for any job $J_i$ and any time $t$, we have 
$f_i(\mathcal{OPT}) \geq p_i \geq q_{t, \mathcal{OPT}}(i)$.
Thus, for any work-conserving schedule $\mathcal{S}$, we have
\begin{align*}
F(\mathcal{OPT}) 
 &=    \sum_{i=1}^{n}{\sum_{t = r_i}^{c_i(\mathcal{OPT})-1}{f_i(\mathcal{OPT})}} 
\geq \sum_{i=1}^{n}{\sum_{t = r_i}^{c_i(\mathcal{OPT})-1}{q_{t, \mathcal{OPT}}(i)}}\\
&=\sum_{t \in \mathbb{N}}{q(t)}
= \sum_{i=1}^{n}{\sum_{t = r_i}^{c_i(\mathcal{S})}{q_{t, \mathcal{S}}(i)}}.
\end{align*}
\end{proof}

To prove Theorem~\ref{thrm: algoN}, we use Lemma~\ref{prop: LBQ} 
to upper bound $\sum_{i=1}^{n}{norm_i(\theta)^2}$. 
Let $I = \{i | norm_i(\theta) = 1, i \in [1,n]\}$.
Obviously, $\sum_{i \in I}{norm_i(\theta)^2} \leq F(\mathcal{OPT})$.
Thus, we only need to consider job indices that are not in $I$.
Let $i \in [1,n] \setminus I$.
Define $T_i = [r_i, t_i-1] = \{r_i, r_i+1, r_i+2, \cdots, t_i-1\}$. 
Because $norm_i(\theta) = t_i-r_i = |T_i|$, 
we have
\begin{equation}\nonumber
norm_i(\theta)^2 = |T_i|^2 = \Theta\left(\sum_{t \in T_i}{(t-r_i)} \right).
\end{equation}
When $t \in T_i$, $\frac{t-r_i}{q_t(i)} < \theta$, 
and thus $t-r_i < \theta q_t(i)$.
As a result, 
$norm_i(\theta)^2 = O(\theta)\sum_{t \in T_i}{q_t(i)}$.
We then have
\[
\sum_{i \in [1, n] \setminus I}{norm_i(\theta)^2} 
= O(\theta)\sum_{i \in [1, n] \setminus I}{\sum_{t \in T_i}{q_t(i)}}
\stackrel{\text{by Lemma~\ref{prop: LBQ}}}{=}   O(\theta)F(\mathcal{OPT}).
\]

%% file: thrmS.tex
\section{Proof of Theorem~\ref{thrm: algoS}}\label{sec: algoS}
Let $u(t)$ be the number of starving tasks at time $t$.  
Thus, by the design of $BAL(\theta)$, 
$c_i(\mathcal{BAL}) \leq t_i+u(t_i)$. 
Therefore, $starv_i(\theta) = c_i(\mathcal{BAL}) - t_i \leq u(t_i)$. 
Define $t^* = \argmax_{t \in \mathbb{N}}{u(t)}$. 
We then have $starv_i(\theta) \leq u(t^*)$ and thus 
$\sum_{i = 1}^{n}{starv_i(\theta)^2} = O(n \cdot u(t^*)^2)$.
As a result, to prove Theorem~\ref{thrm: algoS}, it is sufficient to show
\begin{equation} \label{eq: ultimategoal} \tag{$\ast$}
F(\mathcal{OPT}) = \Omega\left(\sqrt{\theta} \cdot u(t^*)^2\right).
\end{equation}

\subsection{The First Lower Bound of $F(\mathcal{OPT})$}
By Lemma~\ref{prop: LBQ}, we have
\begin{equation} \label{eq: keylb1}
F(\mathcal{OPT}) 
\geq \sum_{i=1}^{n}{\sum_{t = r_i}^{t_i}{q_{t}(i)}}
\geq \sum_{i=1}^{n}{\sum_{t = r_i}^{t_i}{q_{t_i}(i)}} 
\geq \theta \sum_{i=1}^{n}{q_{t_i}(i)^2}, 
\end{equation}
where the last inequality holds because by the design of $BAL(\theta)$, 
$t_i-r_i \geq \theta q_{t_i}(i)$ for any FaS job $J_i$. 
For any FaN job $J_i$, we have $q_{t_i}(i) =0$. 


\subsection{The Second Lower Bound of $F(\mathcal{OPT})$}\label{subsubsec: lb2}
Observe that to minimize $F(\mathcal{S})$ when 
all jobs are released at the same time, 
jobs should be executed in increasing order of their processing times. 
To derive the second lower bound of $F(\mathcal{OPT})$ by this idea, 
jobs that are not active at time $t^*$ under $\mathcal{OPT}$ are ignored (and thus their flow times are 
not counted in $F(\mathcal{OPT})$). For every job $J_i$ that is active at time $t^*$ under $\mathcal{OPT}$, 
we gift $F(\mathcal{OPT})$ by setting $r_i = t^*$ and $p_i = q_{t^*, \mathcal{OPT}}(i)$.
We then execute jobs in increasing order of their new processing times. 
The above idea will be used in Lemma~\ref{lemma: LB2}.


\subsection{Proof of Eq.~\eqref{eq: ultimategoal}}
For any two maps $f$ and $f'$, we write $f$ \textbf{dominates} $f'$, 
or $f'$ is dominated by $f$, if $\dom f' = \dom f$ and $f'(i) \leq f(i), 
\forall i \in \dom f$, where $\dom f$ denotes the domain of $f$. 
We prove the following lemma in Section~\ref{sec: overview}.

\begin{lemma} \label{lemma: LB1}
There is a map $m_{\mathcal{OPT}}$ dominated by $q_{t^*, \mathcal{OPT}}$ such that 
\begin{description}
\item[P1:] $\sum_{i=1}^{n}{m_{\mathcal{OPT}}(i)^2} \leq \sum_{i=1}^{n}{q_{t_i}(i)^2}$, and
\item[P2:] $\sum_{i=1}^{n}{m_{\mathcal{OPT}}(i)} \geq \frac{1}{4} \cdot u(t^*)$.
\end{description}
\end{lemma}

By Eq.~\eqref{eq: keylb1} and P1, we have 
\begin{equation}\label{eq: finalkeylb1}
F(\mathcal{OPT}) \geq \theta \sum_{i=1}^{n}{m_{\mathcal{OPT}}(i)^2}.
\end{equation}
Because $m_{\mathcal{OPT}}$ is dominated by $q_{t^*, \mathcal{OPT}}$, we can use 
the technique developed for the second lower bound to prove the following lemma, whose 
proof is given in Appendix~\ref{appendix: lemma: LB2}. 
\begin{lemma}\label{lemma: LB2}
Let $m_{\mathcal{OPT}}$ be the map defined in Lemma~\ref{lemma: LB1}. 
Reindex jobs so that 
\begin{equation} \label{eq: reindex}
m_{\mathcal{OPT}}(1) \leq m_{\mathcal{OPT}}(2) \leq \cdots \leq m_{\mathcal{OPT}}(n).
\end{equation}
We then have
$\displaystyle 
F(\mathcal{OPT}) \geq
\sum_{i=1}^{n} {\left(\sum_{h=1}^{i} {m_{\mathcal{OPT}}(h)}\right)^2}$.
\end{lemma}

We are now ready to prove Eq.~\eqref{eq: ultimategoal}.
Like Lemma~\ref{lemma: LB2}, we reindex jobs so that 
Eq.~\eqref{eq: reindex} holds.
We then have
\begin{align*}
F(\mathcal{OPT}) 
&= \Omega\left(\sum_{i=1}^{n}{\left(\sum_{h=1}^{i}{m_{\mathcal{OPT}}(h)}\right)^2}
          +\theta \sum_{i=1}^{n}{m_{\mathcal{OPT}}(i)^2}\right) \\
&= \Omega\left(\sqrt{\theta}\right) 
          \sqrt{\left(\sum_{i=1}^{n}{\left(\sum_{h=1}^{i}{m_{\mathcal{OPT}}(h)}\right)^2}\right)
          \left(\sum_{i=1}^{n}{{m_{\mathcal{OPT}}(i)}^2}\right)}\\
&= \Omega\left(\sqrt{\theta}\right)\sum_{i=1}^{n}
                                                 {
                                                   \left(\sum_{h=1}^{i} {m_{\mathcal{OPT}}(h)} \right)            
                                                      m_{\mathcal{OPT}}(i)
                                                 }\\
&= \Omega\left(\sqrt{\theta}\right)\left(\sum_{i=1}^{n}{m_{\mathcal{OPT}}(i)}\right)^2  
= \Omega\left(\sqrt{\theta}\right)u(t^*)^2,
\end{align*}
where the first equality follows from Lemma~\ref{lemma: LB2} and Eq.~\eqref{eq: finalkeylb1}, 
the second equality follows from the AM–GM inequality, 
the third equality follows from the Cauchy-Schwarz inequality, 
and the last equality follows from P2 in Lemma~\ref{lemma: LB1}.

\section{Proof of Lemma~\ref{lemma: LB1}}\label{sec: overview} 
Because Lemma~\ref{lemma: LB1} is trivial when $u(t^*) = 0$, 
we assume $u(t^*) \geq 1$ in the following proof.
We first introduce the following shorthand notations. 
\begin{definition}
For any map $f$, define $S(f) = \sum_{i \in \dom f}{f(i)}$,
and define $f^2$ as a map such that $\dom f^2 = \dom f$ 
and $f^2(x) = f(x)^2$ for any $x \in \dom f^2$.
\end{definition}
Informally, to prove Lemma~\ref{lemma: LB1}, we have to construct a map 
$m_{\mathcal{OPT}}$ dominated by $q_{t^*, \mathcal{OPT}}$ such that 
$S({m_{\mathcal{OPT}}}^2)$ is sufficiently small but
$S(m_{\mathcal{OPT}})$ is sufficiently large. 
Thus, for any $c > 0$, we say that a map $f$ is a 
$c$-\textbf{proper} map if $f$ satisfies the following two constraints:
\begin{enumerate}
\item $S(f^2) \leq \sum_{i=1}^{n}{q_{t_i}(i)^2}$, 
\item $S(f) \geq c \cdot u(t^*)$.
\end{enumerate}

To prove Lemma~\ref{lemma: LB1}, it is sufficient to construct a $\frac{1}{4}$-proper 
map dominated by $q_{t^*, \mathcal{OPT}}$. 
To this end, we first analyze the relationship between 
$q_{t^*, \mathcal{OPT}}$ and $q_{t^*, \mathcal{SRPT}}$. 
In particular, we show that  $q_{t^*, \mathcal{SRPT}}$ \textit{majorizes} $q_{t^*, \mathcal{OPT}}$. 
We then construct another map $h$ based on $q_{t^*}$ such that $h$ majorizes $q_{t^*, \mathcal{SRPT}}$. 
Because majorization is transitive, $h$ majorizes $q_{t^*, \mathcal{OPT}}$. 
We then reduce the task of constructing a $\frac{1}{4}$-proper 
map dominated by $q_{t^*, \mathcal{OPT}}$ to the task of 
constructing a $1$-proper map dominated by $h$. 

\subsection{The Relationship Between 
$q_{t^*, \mathcal{OPT}}$ and $q_{t^*, \mathcal{SRPT}}$} 
Recall that SRPT always executes the active job $J_i$ 
that has the smallest $q_{t, \mathcal{SRPT}}(i)$. 
Thus, SRPT avoids executing the active job that has 
the largest number of remaining tasks. It is not difficult to show that, 
for any positive integer $k$,
among all work-conserving schedules, 
SRPT maximizes the sum of the top-$k$ largest numbers of 
remaining tasks at any time. 
Lemma~\ref{lemma: FRFt} formalizes the above statement using majorization. 
Roughly speaking, if map $f$ majorizes map $g$, 
then for any positive integer $k$, 
the sum of the top-$k$ largest outputs of $f$ is greater than or equal to that of $g$. 

To define majorization, we first introduce the following definition, 
which sorts the domain of a map in decreasing order of their outputs. 
\begin{definition}
Let $f$ be any map. 
Define $\pi_f$ as a function that maps any $k \in [1, |\dom f|]$ to 
the element in $\dom f$ that has the $k$th largest output of $f$ 
(ties can be broken arbitrarily).  
Thus, $\dom f = \{ \pi_f(1), \pi_f(2), \cdots, \pi_f(|\dom f|)\}$
and
$f(\pi_f(1)) \geq f(\pi_f(2)) \geq \cdots \geq f(\pi_f(|\dom f|))$.
\end{definition}
Next, we define the sum of the top-$k$ outputs of a map $f$, denoted by $S_k(f)$.
\begin{definition}
Let $f$ be any map. Define
\begin{equation}\nonumber
  S_k(f)=
    \begin{cases}
      0 & \mbox{if $k =0$}\\
      \sum_{j=1}^{k}{f(\pi_f(j))} & \mbox{if $k \in [1, |\dom f|]$}\\
      S(f) & \mbox{if $k > |\dom f|$}
    \end{cases}       
\end{equation}
\end{definition}
\begin{definition}
For any two maps $f$ and $g$, $f$ \textbf{majorizes} $g$
if the following two conditions are met:
\begin{enumerate}
\item $S(f) = S(g)$, 
\item $S_k(f) \geq S_k(g), \forall k \in \mathbb{N}$.
\end{enumerate}
\end{definition}

The next lemma formalizes the previous discussion on SRPT. 
The proof is based on the simple property that SRPT avoids executing the active jobs that have 
the most remaining tasks and can be found in Appendix~\ref{appendix: lemma: FRFt}. 
\begin{lemma}\label{lemma: FRFt}
Let $\mathcal{S}$ be any work-conserving schedule. 
For any $t \in \mathbb{N}$,
$q_{t, \mathcal{SRPT}}$ majorizes $q_{t, \mathcal{S}}$.
\end{lemma}


\subsection{A Map $h$ That Majorizes $q_{t^*, \mathcal{SRPT}}$ and a $1$-Proper Map Dominated by $h$}
    We begin by defining a map $\hat{q}_{t^*}$ that \emph{freezes} 
    the number of remaining tasks for a job once it becomes starving. 
    For an FaN job (or an FaS job with $t^* \leq t_i$) we set
    \[
    \hat{q}_{t^*}(i) = q_{t^*}(i),
    \]
    while for an FaS job with $t^* > t_i$ we set
    \[
    \hat{q}_{t^*}(i) = q_{t_i}(i).
    \]
    By Line~12 of Algorithm~\ref{algo}, this ensures that for every job $i$,
    \[
    \hat{q}_{t^*}(i) \geq q_{t^*, \mathcal{SRPT}}(i),
    \]
    so that $\hat{q}_{t^*}$ dominates $q_{t^*, \mathcal{SRPT}}$.

    Although $\hat{q}_{t^*}$ dominates $q_{t^*, \mathcal{SRPT}}$, 
    its sum $S(\hat{q}_{t^*})$ may exceed $q(t^*) = S(q_{t^*, \mathcal{SRPT}})$. 
    To address this, we truncate $\hat{q}_{t^*}$ by reducing the smallest nonzero entries of 
    $\hat{q}_{t^*}$ so that $S(\hat{q}_{t^*}) = q(t^*)$
    while preserving the majorization relation. The resulting truncated map is $h$.
    The detailed construction of $h$ can be found in Appendix~\ref{appendix: construction_h}.

    To obtain a $1$-proper map $h'$ dominated by $h$,  
    we set the values of $h$ to zero for all FaN jobs, 
    ensuring that $h'$ satisfies the required properness conditions. 
    The detailed construction of $h'$ can be found in Appendix~\ref{appendix: construction_h'}.

\subsection{The Reduction}
The next lemma shows that 
to construct a $\frac{1}{4}$-proper map dominated by $q_{t^*, \mathcal{OPT}}$, 
it suffices to construct a $1$-proper map dominated by $h$ (e.g., $h'$). 
\begin{lemma}\label{lemma: S1S2} 
Let $f$ and $g$ be such that $f$ majorizes $g$. 
For any $c > 0$, if there is a $c$-proper map dominated by $f$, then 
there is a $\frac{c}{4}$-proper map dominated by $g$.
\end{lemma}

We only give a proof sketch of Lemma~\ref{lemma: S1S2} here. 
The complete proof can be found in Appendix~\ref{appendix: lemma: S1S2}. 
To prove Lemma~\ref{lemma: S1S2}, it suffices to prove that 
for any map $f'$ dominated by $f$, 
there is a map $g'$ dominated by $g$ such that 
$S(g'^2) \leq S(f'^2)$ and $S(g') \geq \frac{1}{4} S(f')$. 
Thus, if $f'$ is a $c$-proper map, then $g'$ is a $\frac{c}{4}$-proper map. 
The construction of $g'$ is purely combinatorial. 
Initially, $g' = g$ and we find a set $I \subseteq \dom f'$ 
such that $\sum_{i \in I}{f'(i)} \geq \frac{1}{2} S(f')$. 
We then construct a family $\{B_i\}_{i \in I}$ of mutually disjoint subsets of $\dom g'$,   
such that for each $i \in I$,   
\[
\frac{f'(i)}{2} \leq \sum_{b \in B_i}{g'(b)} \leq f'(i).
\] 
To this end, we may decrease $g'$. 
For any $b \notin \bigcup_{i \in I}B_i$, we set $g'(b) = 0$. 
$g'$ is then the desired map.

\subsection{Proof of Lemma~\ref{lemma: LB1}}
Because majorization is transitive, $h$ majorizes $q_{t^*, \mathcal{OPT}}$. 
Because there is a $1$-proper map dominated by $h$, 
Lemma~\ref{lemma: S1S2} implies that there is a $\frac{1}{4}$-proper map dominated 
by $q_{t^*, \mathcal{OPT}}$. This completes the proof of Lemma~\ref{lemma: LB1}. 

%% file: CR.tex
\section{Discussion on Standard Scheduling Algorithms}\label{sec: cr}
\subsection{Discussion on RR.}
The lower bound instance of RR in~\cite{doi:10.1137/090772228} shows 
that RR's competitive ratio for minimizing the $\ell_2$ norm of flow time 
is $\Omega(\sqrt{n}/ \log{n})$. The root cause is that RR frequently switches job execution, 
and thus delays the completion of many jobs. 
As a result, RR may have a poor average flow time, which in turn deteriorates 
the $\ell_2$ norm of flow time.

\subsection{Discussion on SRPT.}\label{appendix: SRPT}
\begin{proposition}
Let $\mathcal{S}$ be a schedule that minimizes the average flow time. 
Then $F(\mathcal{S}) \leq n \cdot F(\mathcal{OPT})$.
\end{proposition}
\begin{proof}
\begin{align*}
&n \cdot F(\mathcal{OPT}) = \left(\sum_{i=1}^{n}{1^2}\right) \left( \sum_{i=1}^{n}{{f_i(\mathcal{OPT})}^2}\right)
\geq \left(\sum_{i=1}^{n}{f_i(\mathcal{OPT})}\right)^2 \geq \left(\sum_{i=1}^{n}{f_i(\mathcal{S})}\right)^2 \\
&\geq \sum_{i=1}^{n}{f_i(\mathcal{S})^2} = F(\mathcal{S}),
\end{align*}
where the first inequality follows from the Cauchy-Schwarz inequality and the 
second inequality follows from the assumption that $\mathcal{S}$ minimizes the average flow time.
\end{proof}

It is well-known that SRPT minimizes the average flow time. Thus, we have the following corollary.
\begin{corollary}\label{coro: SRPTcr}
SRPT is $O(\sqrt{n})$-competitive for minimizing the $\ell_2$ norm of flow time.
\end{corollary}
Observe that the instance shown in Fig.~\ref{fig: ex1} not only shows that SRPT may 
cause job starvation, it also shows that the competitive ratio in Corollary~\ref{coro: SRPTcr} 
is asymptotically tight by comparing to FCFS. 

\subsection{Discussion on FCFS.}\label{appendix: FCFS}
\begin{proposition}
Let $\mathcal{S}$ be a schedule that minimizes the maximum flow time. 
Then $F(\mathcal{S}) \leq n \cdot F(\mathcal{OPT})$.
\end{proposition}
\begin{proof}
$$n \cdot F(\mathcal{OPT}) \geq n \max_{i}{{f_i(\mathcal{OPT})}^2} \geq n \max_{i}{{f_i(\mathcal{S})}^2} 
\geq F(\mathcal{S}),$$
where the
second inequality follows from the assumption that $\mathcal{S}$ minimizes the maximum flow time.
\end{proof}

Because FCFS minimizes the maximum flow time~\cite{10.5555/314613.314715}, 
we have the following corollary.
\begin{corollary}\label{coro: FCFScr}
FCFS is $O(\sqrt{n})$-competitive for minimizing the $\ell_2$ norm of flow time.
\end{corollary}
Observe that the instance shown in Fig.~\ref{fig: ex2} not only shows that FCFS may 
deteriorate average flow time, it also shows that 
the competitive ratio in Corollary~\ref{coro: FCFScr} is asymptotically tight by comparing to SRPT. 

\begin{figure}[t]
\begin{center} 
\includegraphics[width=8cm]{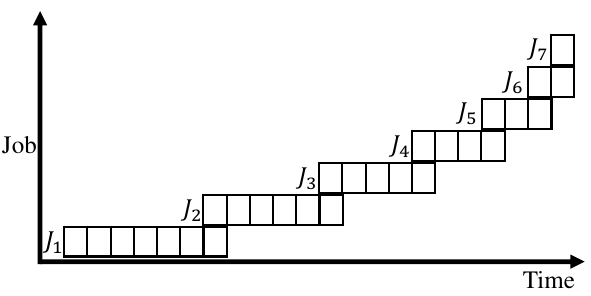} 
\caption{A lower bound instance of SJF and SETF when $n=7$.}
\label{fig: ex3}
\end{center}
\end{figure}

\subsection{Discussion on SJF and SETF.}\label{appendix: SJF}
Under SJF, the server always executes the job that has the smallest job size. 
Under SETF, the server always executes the job that has been been executed the least so far. 
Take Fig.~\ref{fig: ex1} as an example. 
Under SJF, $J_1$ is completed first (assuming that ties are broken by choosing the job with 
the smallest index), and $J_2$ has to wait until all the other jobs are completed. 
Under SETF, both $J_1$ and $J_2$ have to wait until all the other jobs are completed. 
Thus, both SJF and SETF cause job starvation.
By comparing to FCFS, Fig.~\ref{fig: ex1} shows that the competitive ratios of 
SJF and SETF for minimizing the $\ell_2$ norm of flow time are $\Omega(\sqrt{n})$. 
Next, we give an instance to show that the competitive ratios of SJF and SETF are $\Omega(n)$.
In this instance, for every job $J_i$, $p_i = n-i+1$. We set $r_1 = 1$. 
For $i \geq 2$, we set $r_i = r_{i-1}+p_{i-1}-1$. 
Fig.~\ref{fig: ex3} shows such an instance when $n = 7$. 

It is easy to see that under SJF or SETF, every job is completed after $J_n$ is released. 
Thus, for every job $J_i$ with $i \leq \frac{n}{2}$, 
its flow time under SJF or SETF is $\Omega(1+2+3+\cdots+\frac{n}{2})=\Omega(n^2)$. 
As a result, the $\ell_2$ norm of flow time of SJF or SETF is $\Omega(\sqrt{n^5})$. 
It is easy to see that for every job $J_i$, $f_i(\mathcal{FCFS}) = n$. 
Therefore, $F(\mathcal{OPT}) \leq F(\mathcal{FCFS}) = \Theta(n^3)$.
We then have the following result.

\begin{proposition}
The competitive ratios of SJF and SETF for minimizing the $\ell_2$ norm of flow time are $\Omega(n)$.
\end{proposition} 

In hindsight, the root cause of such a poor performance is that SETF and SJF may 
have poor average flow time and poor maximum flow time simultaneously.

%% file: thrmS2.tex
\section{Proof of Theorem~\ref{coro: beat}}\label{appendix: coro: beat}
\begin{align*} 
&F(\mathcal{BAL}(\tilde{n}^{\frac{2}{3}})) 
= \sum_{i=1}^{n}{\left(norm_i(\tilde{n}^{\frac{2}{3}})+starv_i(\tilde{n}^{\frac{2}{3}})\right)^2}\\ 
&= 
O\left(\sum_{i=1}^{n}{norm_i(\tilde{n}^{\frac{2}{3}})^2}
+\sum_{i=1}^{n}{starv_i(\tilde{n}^{\frac{2}{3}})^2}\right)\\
&\stackrel{\text{by Theorems~\ref{thrm: algoN}~and~\ref{thrm: algoS}}}{=} 
O\left(\tilde{n}^{\frac{2}{3}} + \frac{n}{\tilde{n}^{\frac{1}{3}}}\right)F(\mathcal{OPT})
\end{align*}
Thus, the competitive ratio of 
$BAL(\tilde{n}^{\frac{2}{3}})$ for minimizing the $\ell_2$ norm of flow time is 
$O\left(\tilde{n}^{\frac{1}{3}} + n^{\frac{1}{2}}/\tilde{n}^{\frac{1}{6}}\right) 
= O\left(n^{\frac{1}{3}} \left(\alpha^{\frac{1}{3}} + \beta^{\frac{-1}{6}}\right) \right)$. 

\section{Proof of Lemma~\ref{lemma: LB2}} \label{appendix: lemma: LB2}
To derive a lower bound of $F(\mathcal{OPT})$, 
we remove some  tasks from the instance at time $t^*$
so that under $\mathcal{OPT}$, 
every job $J_i$ has exactly $m_{\mathcal{OPT}}(i)$ remaining tasks at time $t^*$. 
The above modification is achievable because $q_{t^*, \mathcal{OPT}}$ dominates $m_{\mathcal{OPT}}$.
Let $i^*$ be the smallest integer such that $m_{\mathcal{OPT}}(i^*) > 0$.
Further assume that starting from time $t^*$, 
the goal of $\mathcal{OPT}$ becomes to minimize
\begin{equation}\label{eq: newgoal}
\sum_{i= i^*}^{n}{(c'_i(\mathcal{OPT})-t^*)^2},
\end{equation}
where $c'_i(\mathcal{OPT})$ is the completion time of $J_i$ under 
$\mathcal{OPT}$ in the modified instance. 
Observe that for all $i \geq i^*$, because $q_{t^*, \mathcal{OPT}}(i) 
\geq m_{\mathcal{OPT}}(i) > 0$, 
$J_i$'s original release time is at most $t^*$.
Thus, Eq.~\eqref{eq: newgoal} is a lower bound of $F(\mathcal{OPT})$.
By Eq.~\eqref{eq: reindex}, to 
minimize Eq.~\eqref{eq: newgoal}, 
$\mathcal{OPT}$ should execute $J_{i^*}, J_{i^*+1}, \cdots, J_{n}$ 
in increasing order of their job indices. 
Thus, for any $i \geq i^*$,
$c_i'(\mathcal{OPT})-t^* = \sum_{h=i^*}^{i}{m_{\mathcal{OPT}}(h)}$, which implies Lemma~\ref{lemma: LB2}.

\input{appendix}

%% file: appendix.tex
\section{Proof of Lemma~\ref{lemma: FRFt}} \label{appendix: lemma: FRFt}
To prove Lemma~\ref{lemma: FRFt}, 
we will consider the restriction of some map $f$ to some subset of $\dom f$. 
Specifically, for any map $f$ and any set $S \subseteq \dom f$, 
the restriction of $f$ to $S$, denoted by $f|_S$, is a map 
from $S$ to $\mathbb{N}$ such that $f|_S(i) = f(i)$ for any $i \in S$. 
The following definition adds an element $x$ and its associated output $y$ to a map $f$.
\begin{definition}
For any map $f$, any $x \notin \dom f$, 
and any $y \in \mathbb{N}$,
define $f \cup (x, y): \dom f \cup \{x\} \rightarrow \mathbb{N}$ 
as a map such that $(f \cup (x, y))(i) = f(i)$ 
if $i \in \dom f$ and $(f \cup (x, y))(x) = y$.
\end{definition}
The following definition considers the union of two disjoint maps.
\begin{definition}
For any two maps $f$ and $h$ such that 
$\dom f \cap \dom h = \varnothing$, 
define $f \cup h: \dom f \cup \dom h \rightarrow \mathbb{N}$ 
as a map such that 
$(f \cup h)|_{\dom f} = f$ and $(f \cup h)|_{\dom h} = h$.
\end{definition}

\begin{example}
Assume $f(1) = 1$, $f(2) = 2$, $h(3) = 3$, and $h(4) = 4$.
Further assume $\dom f = \{1, 2\}$ and $\dom h = \{3, 4\}$.
We then have $\dom (f \cup h) = \{1, 2, 3, 4\}$ and
$(f \cup h)(1) = 1, (f \cup h)(2) = 2, (f \cup h)(3) = 3$, 
and $(f \cup h)(4) = 4$.  
\end{example}

\begin{lemma}\label{lemma: union}
Let $f$ and $g$ be any two maps such that $S_k(f) \geq S_k(g)$ for every 
$k \in \mathbb{N}$. 
Let $x \notin (\dom f \cup \dom g)$.
Then for any $k, y \in \mathbb{N}$, 
$S_k(f \cup (x,y)) \geq S_k(g \cup (x,y))$.
\end{lemma}

\begin{proof}
Assume that in $\dom f$ (respectively, $\dom g$), there are $p_f$ 
(respectively, $p_g$) elements $i$ satisfying $f(i) \geq y$ 
(respectively, $g(i) \geq y$). 

\paragraph*{Case 1: $p_f \leq p_g$.}
	\begin{itemize}
	    \item If $k \leq p_g$, 
	    then $S_k(f \cup (x,y)) \geq S_k(f) 
        \geq S_{k}(g) = S_k(g\cup (x,y))$.
        \item If $k \geq p_g+1$, 
        then $S_k(f \cup (x,y)) 
        = S_{k-1}(f)+y \geq S_{k-1}(g)+y 
        = S_k(g\cup (x,y))$.
	\end{itemize}	 
	
\paragraph*{Case 2: $p_g < p_f$.}
	\begin{itemize}
		\item If $k \leq p_g$,
		then $S_k(f\cup (x,y)) = S_k(f) 
		\geq S_k(g) = S_k(g \cup (x,y))$.
		\item If $p_g+1 \leq k \leq p_f$,
		then $S_k(f\cup (x,y)) 
        = S_k(f) 
        \geq S_{p_g}(f)+y(k-p_g)
        \geq S_{p_g}(g)+y(k-p_g) 
        \geq S_k(g \cup (x,y))$.
        \item If $k \geq p_f+1$,
        then $S_k(f\cup (x,y)) = S_{k-1}(f)+y 
		\geq S_{k-1}(g)+y = S_k(g \cup (x,y))$.
	\end{itemize}
\end{proof}

\begin{lemma}\label{coro: topk}
Let $f$, $g$, and $h$ be any three maps such that 
$\dom f \cap \dom h = \varnothing$, 
$\dom g \cap \dom h = \varnothing$, 
and $S_k(f) \geq S_k(g), \forall k \in \mathbb{N}$.
Then $S_k(f \cup h) \geq S_k(g \cup h), \forall k \in \mathbb{N}$.
\end{lemma}

\begin{proof}
We prove Lemma~\ref{coro: topk} by induction on $|\dom h|$.
When $|\dom h| = 1$, the lemma holds due to Lemma~\ref{lemma: union}.
Assume the lemma holds when $|\dom h| = z$.
When $|\dom h| = z+1$, pick any $x \in \dom h$. 
Consider the map $h' = h|_{\dom h \setminus \{x\}}$.
Thus, $|\dom h'| = z$ and by the induction hypothesis, we then have
$S_k(f \cup h') \geq S_k(g \cup h')$ for any $k \in \mathbb{N}$.
By Lemma~\ref{lemma: union}, we then have 
$S_k((f \cup h') \cup (x, h(x))) \geq S_k((g \cup h') \cup (x, h(x)))$
for any $k \in \mathbb{N}$.
The proof then follows from $(f \cup h') \cup (x, h(x)) = f \cup h$ and 
$(g \cup h') \cup (x, h(x)) = g \cup h$.
\end{proof}

We are now ready to prove Lemma~\ref{lemma: FRFt}. 
The proof is by induction on $t$.
Clearly, the lemma holds when $t = 0$. 
Assume that the lemma holds when $t = \tau$. 
Thus, 
\[
S_k(q_{\tau, \mathcal{SRPT}}|_{A(\tau, \mathcal{SRPT})}) \geq 
S_k(q_{\tau, \mathcal{S}}|_{A(\tau, \mathcal{S})}),
\forall k \in \mathbb{N}.
\]
Because SRPT always executes the job with the least remaining tasks, 
and $\mathcal{S}$ is work-conserving, 
we then have 
\[
S_k(q_{\tau+1, \mathcal{SRPT}}|_{A(\tau, \mathcal{SRPT})}) \geq 
S_k(q_{\tau+1, \mathcal{S}}|_{A(\tau, \mathcal{S})}),
\forall k \in \mathbb{N}.
\]
Let $I_{\tau+1}$ be the index set of the jobs released at time $\tau+1$.
By Lemma~\ref{coro: topk}, for any $k \in \mathbb{N}$, we have 
\[
S_k(q_{\tau+1, \mathcal{SRPT}}|_{A(\tau, \mathcal{SRPT}) \cup I_{\tau+1}}) \geq 
S_k(q_{\tau+1, \mathcal{S}}|_{A(\tau, \mathcal{S}) \cup I_{\tau+1}}), 
\]
which implies 
\[
S_k(q_{\tau+1, \mathcal{SRPT}}) \geq 
S_k(q_{\tau+1, \mathcal{S}}),
\forall k \in \mathbb{N}. 
\]
Finally, because both $\mathcal{SRPT}$ and $\mathcal{S}$ are 
work-conserving, we have $S(q_{\tau+1, \mathcal{SRPT}}) = 
S(q_{\tau+1, \mathcal{S}})$, which completes the proof.

\begin{remark}
The notion of majorization is first studied by Hardy et al.~\cite{hardy}. 
Golovin et al. use majorization to study all symmetric norms 
of flow time~\cite{golovin_et_al:LIPIcs:2008:1753}. 
The original definition of majorization deals with vectors instead of maps. 
For example, in~\cite{golovin_et_al:LIPIcs:2008:1753}, the notion 
of majorization is applied to vectors consisting of the flow time of all jobs. 
In this paper, we do not directly consider the flow time. 
Instead, we consider the number of remaining tasks of active jobs. 
In addition, we consider the restriction of a map to some subset of $[1,n]$ 
(e.g., $A(\tau, \mathcal{SRPT})$). 
Thus, we use maps instead of vectors.
\end{remark}

\section{Construction of $h$}\label{appendix: construction_h}
\subsection{Freezing $q_{t^*}$ to Dominate $q_{t^*, \mathcal{SRPT}}$}
Consider a map $\hat{q}_t$ that freezes the number of remaining tasks 
of an FaS job once it becomes starving. 
Specifically, for any $t \in \mathbb{N}$ and $i \in [1,n]$, 
define 
\[
\hat{q}_t(i) = 
\begin{cases}
    q_t(i) &\parbox[t]{5.5cm}{if ($J_i$ is an FaN job) \\or ($J_i$ is an FaS job and $t \leq t_i$)}\\
    q_{t_i}(i) &\mbox{if $J_i$ is an FaS job and $t > t_i$.}
\end{cases}
\]
We stress that the definition of $\hat{q}_t(i)$ can be applied to 
$t > c_i(\mathcal{BAL})$. 
The following result is due to Line 12 of Algorithm~\ref{algo}.
\begin{lemma}\label{lemma:algfrf}
For any $t \in \mathbb{N}$,
$\hat{q}_t$ dominates $q_{t, \mathcal{SRPT}}$. 
\end{lemma}

\begin{proof}
It suffices to show that for all $i \in [1, n]$, $q_{t, \mathcal{SRPT}}(i) \leq \hat{q}_t(i)$.
\paragraph*{Case 1: $J_i$ is an FaN job or $J_i$ is an FaS job with $t \leq t_i$.}
In this case, by Line 12 of Algorithm~\ref{algo}, 
for any $t' < t$, if $\mathcal{BAL}$ executes a task $J_{i,k}$ in time slot $[t']$,  
then it must be the case that $\mathcal{SRPT}$ executes $J_{i,k}$ in some
time slot $[t'']$ with $t'' \leq t'$.
Thus, $q_{t, \mathcal{SRPT}}(i) \leq q_t(i) = \hat{q}_t(i)$.
\paragraph*{Case 2: $J_i$ is an FaS job with $t > t_i$.}
In this case, we have 
$\hat{q}_t(i) 
=    q_{t_i}(i) 
\geq q_{t_i, \mathcal{SRPT}}(i) 
\geq q_{t,   \mathcal{SRPT}}(i)$.
\end{proof}

\subsection{Truncating $\hat{q}_{t^*}$ to Majorize $q_{t^*, \mathcal{SRPT}}$}\label{subsubsec: TR}
Because $S(\hat{q}_{t^*})$ may be greater than $q(t^*) = S(q_{t^*, \mathcal{SRPT}})$, $\hat{q}_{t^*}$ 
may not majorize $q_{t^*, \mathcal{SRPT}}$. 
Thus, we construct a map $h$ by decreasing the smallest non-zero outputs of $\hat{q}_{t^*}$ 
so that $S(h) = q(t^*)$. 
We call this operation the truncation of $\hat{q}_{t^*}$ at $q(t^*)$. 

\begin{definition}\label{defi: truncate}
Let $f$ be any map such that $S(f) \geq 1$. 
Let $c$ be any integer in $[1, S(f)]$.
Let $lu(f,c)$ be the least integer such that $S_{lu(f,c)}(f) \geq c$.
The \textbf{truncation} of $f$ at $c$, denoted by $TR[f, c]$, 
is a map with domain $\dom f$ such that
\begin{align*}
  TR[f, c](\pi_f(k))
  =
    \begin{cases}
      f(\pi_f(k)) & \mbox{if $k \in [1, lu(f,c)-1]$}\\
      c - S_{k-1}(f) & \mbox{if $k = lu(f,c)$}\\
      0 & \mbox{if $k \in [lu(f,c)+1, |\dom f|]$}
    \end{cases}       
\end{align*}
Finally, $TR[f, c]$ is said to be a valid truncation 
if $c \in [1, S(f)]$.
\end{definition}

\begin{example}
Assume $f(i) = 10i$ and $\dom f = [1, 10]$.
We then have $S_3(f) = 100+90+80 = 270$ and
$S_4(f) = 100+90+80+70 = 340$.
If $c = 300$, then $lu(f,c) = 4$, 
$TR[f, c](10) = 100, TR[f, c](9) = 90, 
TR[f, c](8) = 80$, and $TR[f, c](7) = 30$.
For all $i \in [1, 6]$, $TR[f, c](i) = 0$.
\end{example}

Clearly, we have the following result.
\begin{lemma}\label{fact: propTR}
Let $TR[f, c]$ be a valid truncation. 
Then $S(TR[f, c]) = c$ and $TR[f, c]$ is dominated by $f$. 
\end{lemma}

\begin{lemma} \label{lemma: TR123} 
Let $h = TR[\hat{q}_{t^*}, q(t^*)]$.
Then $h$ majorizes $q_{t^*, \mathcal{SRPT}}$.
\end{lemma}
\begin{proof}
First note that because $q(t^*) \geq u(t^*) \geq 1$ and $q(t^*) \leq S(\hat{q}_{t^*})$, 
$TR[\hat{q}_{t^*}, q(t^*)]$ is a valid truncation. 
By Lemma~\ref{fact: propTR}, 
$S(TR[\hat{q}_{t^*}, q(t^*)]) = q(t^*) = S(q_{t^*, \mathcal{SRPT}})$. 
It is then sufficient to prove 
\begin{equation}\label{eq: TRSMgoal}
     S_k(TR[\hat{q}_{t^*}, q(t^*)]) 
\geq S_k(q_{t^*, \mathcal{SRPT}}), \forall k \in \mathbb{N}.
\end{equation} 
To prove Eq.~\eqref{eq: TRSMgoal}, 
we consider the following two cases.

\paragraph*{Case 1: $k \leq lu(\hat{q}_{t^*},q(t^*))-1$.}
The case where $k = 0$ is trivial. Thus, we assume 
$1 \leq k \leq lu(\hat{q}_{t^*},q(t^*))-1$.
By the definition of truncation, we have 
\begin{equation}\nonumber
S_k(TR[\hat{q}_{t^*}, q(t^*)]) = S_k(\hat{q}_{t^*}), 
\forall k \in [1, lu(\hat{q}_{t^*},q(t^*))-1].
\end{equation}
The proof then follows from Lemma~\ref{lemma:algfrf}. 
Specifically, by Lemma~\ref{lemma:algfrf}, we have 
\begin{equation}\nonumber
S_k(\hat{q}_{t^*}) \geq S_k(q_{t^*, \mathcal{SRPT}}), 
\forall k \in [1, n].
\end{equation}

\paragraph*{Case 2: $k \geq lu(\hat{q}_{t^*},q(t^*))$.}
In this case, we have
$S_k(TR[\hat{q}_{t^*}, q(t^*)]) = q(t^*) 
\geq S_k(q_{t^*, \mathcal{SRPT}})$.
\end{proof}
 
\section{Construction of $h'$} \label{appendix: construction_h'} 

\begin{lemma} \label{lemma: mbal}
Let $h = TR[\hat{q}_{t^*}, q(t^*)]$. 
There is a $1$-proper map $h'$ dominated by $h$.
\end{lemma}
\begin{proof}
Define 
\[
I_S = \{i | \text{$J_i$ is an FaS job and $t_i \leq t^*$}\}
\] 
and 
\[
I_N = \{i | \text{$J_i$ is an FaN job or $t_i > t^*$}\}.
\] 
Observe that for any $i \in I_N$, we have $\hat{q}_{t^*}(i) = q_{t^*}(i)$, 
and for any $i \in I_S$, we have $\hat{q}_{t^*}(i) = q_{t_i}(i)$.
Define
\begin{equation}\nonumber
  h'(i)=
    \begin{cases}
      h(i) &\mbox{if } i \in I_S\\
      0                             &\mbox{if } i \in I_N.\\
    \end{cases}       
\end{equation} 
Clearly, $h'$ is dominated by $h$.  
In addition, 
\[
    S(h'^2) 
=    \sum_{i \in I_S}{h(i)^2}
\leq \sum_{i \in I_S}{\hat{q}_{t^*}(i)^2}
=    \sum_{i \in I_S}{q_{t_i}(i)^2}
\leq \sum_{i=1}^{n}{q_{t_i}(i)^2}.
\]
Observe that the number of remaining normal tasks at time $t^*$ is $\sum_{i \in I_N}{q_{t^*}(i)}$. 
Thus, $\sum_{i \in I_N}{q_{t^*}(i)} + u(t^*) = q(t^*)$. As a result, 
\begin{align*}
S(h') &=      S(h)       - \sum_{i \in I_N}{h(i)} 
      \geq   S(h)       - \sum_{i \in I_N}{\hat{q}_{t^*}(i)} 
      =    S(h)       - \sum_{i \in I_N}{q_{t^*}(i)} \\
      &=      q(t^*)     - \sum_{i \in I_N}{q_{t^*}(i)} 
       =      u(t^*).
\end{align*}
Therefore, $h'$ is a $1$-proper map dominated by $h$.
\end{proof}

\section{Proof of Lemma~\ref{lemma: S1S2}}\label{appendix: lemma: S1S2}
We will prove the following more general result.
\begin{lemma}\label{lemma: _S1S2}
Let $f$ and $g$ be such that $f$ majorizes $g$. 
For any map $f'$ dominated by $f$, 
there is a map $g'$ dominated by $g$ such that $S(g'^2) \leq S(f'^2)$ 
and $S(g') \geq \frac{1}{4}S(f')$.
\end{lemma}

For any map $f$, we use $f^+$ to denote the restriction 
of $f$ to $\{i| f(i) > 0\}$.
Let $d_{f} = |\dom f^+|$ and $d_{g} = |\dom g^+|$.
Because $f$ majorizes $g$, we have 
\begin{equation}\nonumber
S_{d_g}(g^+) 
= S(g^+)=S(f^+)>  
S_{d_{f}-1}(f^+)
\geq S_{d_{f}-1}(g^+).
\end{equation}
Therefore, $d_g > d_f-1$ and thus  
\begin{equation}\label{eq: dgdf}
d_{g} \geq d_{f}. 
\end{equation}
We write $\dom f$ as $\{a_1, a_2, \cdots, a_{|\dom f|}\}$ such that 
$f(a_1) \geq f(a_2) \geq f(a_3) \geq f(a_{|\dom f|})$. 
Similarly, we write $\dom g$ as $\{b_1, b_2, \cdots, b_{|\dom g|}\}$ such that 
$g(b_1) \geq g(b_2) \geq g(b_3) \geq g(b_{|\dom g|})$. 
We divide $\dom f$ into three sets, $I_0$, $I_1$, and $I_2$,
where 
\begin{align*}
I_0 &= \dom f \setminus \dom f^+ = \{a_k| k \in [d_f+1, |\dom f|]\},\\
I_1 &= \{a_k | f'(a_k) \leq g(b_k), k \in [1, d_{f}]\}, \text{ and }\\
I_2 &= \{a_k | f'(a_k) >    g(b_k), k \in [1, d_{f}]\}.
\end{align*}
Note that by Eq.~\eqref{eq: dgdf}, $b_k$ exists for any $k \in [1, d_{f}]$.
The proof proceeds as follows: 
For $I_1$ (respectively, $I_2$), we will construct a map $g_1$ (respectively, $g_2$) that is dominated by $g$. 
If $\sum_{i \in I_1}{f'(i)} \geq \sum_{i \in I_2}{f'(i)}$, we set $g' = g_1$.
Otherwise, we set $g'=g_2$. 

\subsection{Construction and Properties of $g_1$} 
If $I_1 \neq \varnothing$, we construct a map $g_1$ dominated by $g$.
Initially, $g_1(i) = 0$ for all $i \in \dom g$.
For each $a_k \in I_1$, set $g_1(b_k)$ as $f'(a_k)$. 
Thus, $g_1$ is dominated by $g$.
In addition, we have 
\begin{equation}\label{eq: m1B}
S(g_1^2) = \sum_{a_k \in I_1}{f'(a_k)^2}
\end{equation}
end
\begin{equation}\label{eq: m1A}
S(g_1) = \sum_{a_k \in I_1}{f'(a_k)}.
\end{equation}

\subsection{Construction of $g_2$} 
If $I_2 \neq \varnothing$, we construct another map $g_2$ dominated by $g$.
Initially, $g_2(i) = 0$ for any $i \in \dom g$. 
We rewrite $I_2$ as $\{a_{\kappa(1)}, a_{\kappa(2)}, \cdots, a_{\kappa(|I_2|)}\}$ so that 
$\kappa(1) \leq \kappa(2) \leq \cdots \leq \kappa(|I_2|)$. 
For brevity, for any positive integers $x$ and $y$ with 
$x \leq y \leq d_g$, define 
\begin{equation}\nonumber
S_{x,y}(g) = \sum_{k \in [x,y]}{g(b_k)}.
\end{equation}
Define $y(0) = \kappa(1)-1$.
The construction of $g_2$ proceeds in rounds.
In the $j$th ($j \in [1, |I_2|]$) round, we set 
\begin{equation} \label{eq: xim2}
x(j) = \max{(\kappa(j), y(j-1)+1)}.
\end{equation} 
We set $y(j)$ to be the smallest integer such that 
\begin{equation} \label{eq: yim2}
S_{x(j),y(j)}(g) \geq \frac{f'(a_{\kappa(j)})}{2}.
\end{equation}
We then set 
\begin{equation} \label{eq: m2setting}
g_2(b_k) = g(b_k), \forall k \in [x(j), y(j)]. 
\end{equation}

\begin{figure}[t]
\begin{center} 
\includegraphics[width=8cm]{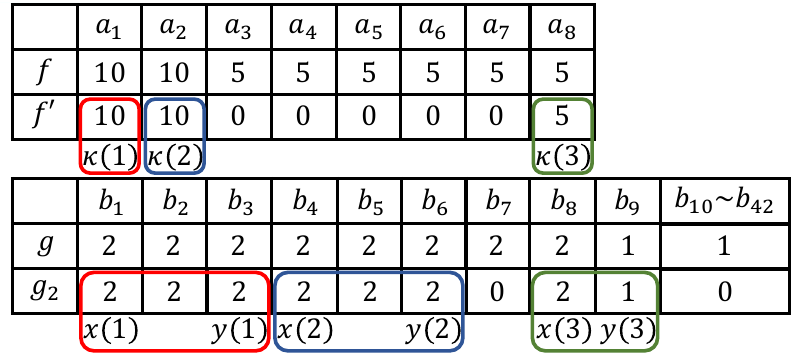} 
\caption{An example of $g_2$.}
\label{fig: exmajor}
\end{center}
\end{figure}

\begin{example}
Consider the maps $f$, $f'$, and $g$ shown in Fig.~\ref{fig: exmajor}. 
We then have $I_2 = \{a_1, a_2, a_8\}$. Thus, $\kappa(1) = 1, \kappa(2) = 2, \kappa(3) = 8$. 
\begin{enumerate}
\item In Round 1, we set $x(1) = \kappa(1) = 1$ and $y(1) = 3$. 
Observe that $f'(a_{\kappa(1)}) = 10$ and $S_{x(1),y(1)}(g) = 6$. Thus, 
\[
\frac{f'(a_{\kappa(1)})}{2} \leq S_{x(1),y(1)}(g) < f'(a_{\kappa(1)}).
\]
\item In Round 2, we set $x(2) = y(1)+1 = 4$ and $y(2) = 6$.
Observe that $f'(a_{\kappa(2)}) = 10$ and $S_{x(2),y(2)}(g) = 6$. Thus,
\[
\frac{f'(a_{\kappa(2)})}{2} \leq S_{x(2),y(2)}(g) < f'(a_{\kappa(2)}).
\]
\item In Round 3, we set $x(3) = \kappa(3) = 8$ and $y(3) = 9$.
Observe that $f'(a_{\kappa(3)}) = 5$ and $S_{x(3),y(3)}(g) = 3$. Thus, 
\[
\frac{f'(a_{\kappa(3)})}{2} \leq S_{x(3),y(3)}(g) < f'(a_{\kappa(3)}).
\]
\end{enumerate}
\end{example}

\subsection{Correctness of the Construction of $g_2$}
To prove the above construction is correct, it suffices to prove that in every round $j$, 
$x(j)$ is valid (i.e., $x(j) \leq d_g$) and 
$y(j)$ is valid (i.e., there exists $y(j)$ that satisfies Eq.~\eqref{eq: yim2}). 
We first prove the following lemma.
\begin{lemma}\label{cl:xyleqf}
If $x(j)$ and $y(j)$ are valid,  
then $S_{x(j), y(j)}(g) < f'(a_{\kappa(j)})$.
\end{lemma}
\begin{proof} 
First, we have
\begin{equation}\nonumber
f'(a_{\kappa(j)})
\stackrel{\text{by $a_{\kappa(j)} \in I_2$}}{>} 
g(b_{\kappa(j)})
\stackrel{\text{by Eq.~\eqref{eq: xim2}}}{\geq}
g(b_{x(j)}).
\end{equation}
Thus, $S_{x(j), x(j)}(g) < f'(a_{\kappa(j)})$.
Moreover, for any $h \in [1, |\dom g|-1]$, we have $g(b_{h}) \geq g(b_{h+1})$.
Thus, if $S_{x(j), h}(g) < f'(a_{\kappa(j)})/2$, then $S_{x(j), h+1}(g) < f'(a_{\kappa(j)})$. 
The proof then follows from the definition of $y(j)$.
\end{proof}

Next, we prove that the following statements 
$\mathcal{X}(j)$ and $\mathcal{Y}(j)$ hold for any 
$j \in [1, |I_2|]$ by induction on $j$:
\begin{description}
\item[$\mathcal{X}(j)$:] $x(j) \leq d_g$.
\item[$\mathcal{Y}(j)$:] $S_{x(j), d_g}(g) \geq \sum_{h=j}^{|I_2|}{f'(a_{\kappa(h)})}$.
\end{description}
Observe that $\mathcal{X}(j)$ and $\mathcal{Y}(j)$ imply that 
$x(j)$ and $y(j)$ are valid, respectively.

When $j = 1$, $x(1) = \kappa(1) \leq d_f \leq d_g$. Thus, $\mathcal{X}(1)$ 
holds. In addition, 
\begin{align*}
&S_{x(1), d_g}(g) = S_{\kappa(1), d_g}(g) = 
S(g) - S_{\kappa(1)-1}(g)
\geq
S(f) - S_{\kappa(1)-1}(f)\\
&=    \sum_{h = \kappa(1)}^{d_f}{f(a_h)}
\geq \sum_{h = \kappa(1)}^{d_f}{f'(a_h)} 
\geq \sum_{h=1}^{|I_2|}{f'(a_{\kappa(h)})}.
\end{align*}
Thus, $\mathcal{Y}(1)$ holds.

Assume $\mathcal{X}(j)$ and $\mathcal{Y}(j)$ hold when $j = z$ 
for some $z \in [1, |I_2|-1]$. 
To prove $\mathcal{X}(z+1)$ and $\mathcal{Y}(z+1)$ hold, 
we first consider the case where $\kappa(z+1) \geq y(z)+1$.
In this case, $x(z+1) = \kappa(z+1) \leq d_f \leq d_g$. 
Thus, $\mathcal{X}(z+1)$ holds.
In addition, 
\begin{align*}
&S_{x(z+1), d_g}(g) = S_{\kappa(z+1), d_g}(g) = S(g) - S_{\kappa(z+1)-1}(g)
\geq
S(f) - S_{\kappa(z+1)-1}(f)\\
&=    \sum_{h = \kappa(z+1)}^{d_f}{f(a_h)} 
\geq \sum_{h = \kappa(z+1)}^{d_f}{f'(a_h)}
\geq \sum_{h=z+1}^{|I_2|}{f'(a_{\kappa(h)})}.
\end{align*}
Thus, $\mathcal{Y}(z+1)$ holds.

Next, we consider the case where 
$y(z)+1 > \kappa(z+1)$. Thus, $x(z+1) = y(z)+1$. 
We have
\begin{align}
\nonumber      &S_{x(z), d_g}(g)-S_{x(z), y(z)}(g) 
\stackrel{\text{by $\mathcal{Y}(z)$}}{\geq} 
\sum_{h=z}^{|I_2|}{f'(a_{\kappa(h)})} - S_{x(z), y(z)}(g)\\
\label{eq: xyvalid} 
&\stackrel{\text{by Lemma~\ref{cl:xyleqf}}}{>} 
  \sum_{h=z}^{|I_2|}{f'(a_{\kappa(h)})} - f'(a_{\kappa(z)})
= \sum_{h=z+1}^{|I_2|}{f'(a_{\kappa(h)})}.
\end{align}
Because $a_{\kappa(z+1)} \in I_2$, $f'(a_{\kappa(z+1)}) > 0$. 
By Eq.~\eqref{eq: xyvalid}, 
$S_{x(z), d_g}(g)-S_{x(z), y(z)}(g) > 0$.
As a result, $d_g \geq y(z)+1=x(z+1)$ and 
\[
S_{x(z+1), d_g}(g) = S_{y(z)+1, d_g}(g) = S_{x(z), d_g}(g)-S_{x(z), y(z)}(g)
\stackrel{\text{by Eq.~\eqref{eq: xyvalid}}}{>} 
\sum_{h=z+1}^{|I_2|}{f'(a_{\kappa(h)})}.
\]
Thus, both $\mathcal{X}(z+1)$ and $\mathcal{Y}(z+1)$ hold.
By mathematical induction, $\mathcal{X}(j)$ and $\mathcal{Y}(j)$ hold 
for any $j \in [1, |I_2|]$.

\subsection{Properties of $g_2$}
Clearly, $g_2$ is dominated by $g$. 
Because we have $g_2(b_k) > 0$ only when $k \in [x(j), y(j)]$ 
for some $j \in [1, |I_2|]$,
we then have
\begin{equation}\label{eq: m2B}
S(g_2^2)
\leq
\sum_{j=1}^{|I_2|}{\left(\sum_{k \in [x(j), y(j)]}{g_2(b_k)}\right)^2} 
\stackrel{\text{by Eq.~\eqref{eq: m2setting}}}{=}
\sum_{j=1}^{|I_2|}{S_{x(j), y(j)}(g)^2} 
\stackrel{\text{by Lemma~\ref{cl:xyleqf}}}{<}
\sum_{i \in I_2}{f'(i)^2}. 
\end{equation}
In addition, by Eq.~\eqref{eq: xim2}, $[x(j), y(j)] \cap [x(j'), y(j')] = \varnothing$ 
if $j \neq j'$.
We then have
\begin{equation}\label{eq: m2A}
S(g_2) = \sum_{j=1}^{|I_2|}{\sum_{k \in [x(j), y(j)]}{g_2(b_k)}} 
 \stackrel{\text{by Eq.~\eqref{eq: m2setting} and ~\eqref{eq: yim2}}}{\geq} 
\sum_{j=1}^{|I_2|}{\frac{f'(a_{\kappa(j)})}{2}}
= \frac{1}{2}\sum_{i \in I_2}{f'(i)}. 
\end{equation}

\subsection{Construction of $g'$ Based on $g_1$ and $g_2$}
If $\sum_{i \in I_1}{f'(i)} \geq \sum_{i \in I_2}{f'(i)}$, we set $g' = g_1$.
Otherwise, we set $g'=g_2$. 
Because both $g_1$ and $g_2$ are dominated by $g$, $g'$ is dominated by $g$.
Observe that $S(f') = \sum_{i \in I_1}{f'(i)} + \sum_{i \in I_2}{f'(i)}$.
Thus, if $g' = g_1$, by Eq.~\eqref{eq: m1A}, we have 
\[
S(g') = S(g_1) = \sum_{i \in I_1}{f'(i)} \geq \frac{1}{2}S(f').
\] 
By Eq.~\eqref{eq: m1B},
we have 
\[
S(g'^2) = S(g_1^2) = \sum_{i \in I_1}{f'(i)^2} \leq S(f'^2).
\]
Finally, if $g' = g_2$, by Eq.~\eqref{eq: m2A}, we have 
\[
S(g') = S(g_2) \geq \frac{1}{2}\sum_{i \in I_2}{f'(i)} \geq \frac{1}{4}S(f').
\]
By Eq.~\eqref{eq: m2B}, we have 
\[
S(g'^2) = S(g^2_2) < \sum_{i \in I_2}{f'(i)^2} \leq S(f'^2).
\]